\documentclass[conference]{IEEEtran}
\IEEEoverridecommandlockouts
\makeatletter
\def\ps@headings{%
	\def\@oddhead{\mbox{}\scriptsize\rightmark \hfil \thepage}%
	\def\@evenhead{\scriptsize\thepage \hfil \leftmark\mbox{}}%
	\def\@oddfoot{}%
	\def\@evenfoot{}}
\makeatother
\usepackage{nicefrac}
\usepackage[para]{threeparttable}
\usepackage{graphicx}
\usepackage{algorithm}
\usepackage[noend]{algpseudocode}

\usepackage{subcaption}
\usepackage{url}
\usepackage{graphicx}

\usepackage{epsfig}
\usepackage{url}
\usepackage{stfloats}
\usepackage{enumerate}
\usepackage{longtable}

\usepackage{latexsym}
\usepackage{verbatim}
\usepackage{amssymb}

\usepackage{xspace}

\usepackage{xcolor}

\usepackage{verbatim}
\usepackage{caption}
\usepackage{subcaption}
\usepackage{amsmath}
\usepackage{amssymb}
\usepackage{wrapfig}
\usepackage{tabularx}
\usepackage{lscape}
\usepackage{float}
\usepackage{inputenc}
\usepackage[utf8]{luainputenc}
\usepackage[T1]{fontenc}
\usepackage{pgfgantt}
\usepackage{footnote}

\usepackage{algpseudocode}
\usepackage{algorithm}
\usepackage{multicol}
\usepackage{multirow}
\usepackage{tablefootnote}
\usepackage{mathtools}
\usepackage{graphicx}
\usepackage{caption}
\usepackage{subcaption}
\usepackage{amsmath}
\usepackage{tikz}
\usepackage{float}
\usepackage{bm}
\usepackage{url}
\usepackage{amsthm}

\newcommand{\ignore}[1]{}
\usetikzlibrary{matrix,shapes,arrows,positioning,chains, calc}

\newcommand{\specialcell}[2][c]{%
	\begin{tabular}[#1]{@{}c@{}}#2\end{tabular}}

\mathchardef\mhyphen="2D

\usepackage{pbox}
\newtheorem{definition}{Definition}

\newcommand{\algrule}[1][.2pt]{\par\vskip.5\baselineskip\hrule height #1\par\vskip.5\baselineskip}

\newcommand{\A}{$\mathcal{A}$}
\newcommand{\F}{$\mathcal{F}$}
\newcommand{\Rq}{\ensuremath \stackrel{\$}{\leftarrow}\mathbb{Z}_{q}^{*}{\xspace}}

\mathchardef\mhyphen="2D
\newcommand{\sk}{\ensuremath {\mathit{sk}}{\xspace}}
\newcommand{\pk}{\ensuremath {\mathit{PK}}{\xspace}}
\newcommand{\as}{\ensuremath {\leftarrow}{\xspace}}

\newcommand{\lhi}{\ensuremath {\mathcal{HL}_i}{\xspace}}
\newcommand{\lh}{\ensuremath {\mathcal{HL}_0}{\xspace}}
\newcommand{\lhp}{\ensuremath {\mathcal{HL}_1}{\xspace}}
\newcommand{\lhpp}{\ensuremath {\mathcal{HL}_2}{\xspace}}
\newcommand{\lm}{\ensuremath {\mathcal{LM}}{\xspace}}
\newcommand{\lw}{\ensuremath {\mathcal{LW}}{\xspace}}

\newcommand{\any}{\ensuremath { (\cdot)}{\xspace}}

\newcommand{\hsim}{\ensuremath {\mathit{H}\mhyphen\mathit{Sim}}{\xspace}}

\newcommand{\sgn}{\ensuremath {\mathit{SGN}}{\xspace}}
\newcommand{\sgnkg}{\ensuremath {\mathit{SGN.Kg}}{\xspace}}
\newcommand{\sgnsig}{\ensuremath {\mathit{SGN.Sig}}{\xspace}}
\newcommand{\sgnver}{\ensuremath {\mathit{SGN.Ver}}{\xspace}}
\newcommand{\EUCMA}{\ensuremath {\mathit{EU}\mhyphen\mathit{CMA}}{\xspace}}

\newcommand{\ro}{\ensuremath {\mathit{RO}(\cdot)}{\xspace}}

\newcommand{\advsch}{\ensuremath {\mathit{Adv}_{\sch}^{\EUCMA}(t',q'_s)}{\xspace}}

{\xspace}

\newcommand{\sch}{\ensuremath {\texttt{Schnorr}}{\xspace}}
\newcommand{\schkg}{\ensuremath {\texttt{Schnorr.Kg}}{\xspace}}
\newcommand{\schsig}{\ensuremath {\texttt{Schnorr.Sig}}{\xspace}}
\newcommand{\schver}{\ensuremath {\texttt{Schnorr.Ver}}{\xspace}}

\newcommand{\params}{\ensuremath {\mathit{params}}{\xspace}}

\newcommand{\qs}{\ensuremath {\mathit{q_s}}{\xspace}}

\newcommand{\ulsv}{\ensuremath {\texttt{ESEM$_2$}}{\xspace}}
\newcommand{\uls}{\ensuremath {\texttt{ESEM}}{\xspace}}
\newcommand{\ulskg}{\ensuremath {\texttt{ESEM.Kg}}{\xspace}}
\newcommand{\ulssig}{\ensuremath {\texttt{ESEM.Sig}}{\xspace}}
\newcommand{\ulsver}{\ensuremath {\texttt{ESEM.Ver}}{\xspace}}

\newcommand{\node}{\ensuremath {\texttt{SNOD$\mhyphen$BPV}}{\xspace}}
\newcommand{\nodekg}{\ensuremath {\texttt{SNOD$\mhyphen$BPV.Offline}}{\xspace}}
\newcommand{\nodesig}{\ensuremath {\texttt{SNOD$\mhyphen$BPV.Sender}}{\xspace}}
\newcommand{\nodever}{\ensuremath {\texttt{SNOD$\mhyphen$BPV.Receiver}}{\xspace}}

\newcommand{\BPV}{\ensuremath {\mathit{BPV}}{\xspace}}
\newcommand{\BPVOff}{\ensuremath {\mathit{BPV.Offline}}{\xspace}}
\newcommand{\BPVOn}{\ensuremath {\mathit{BPV.Online}}{\xspace}}

\newcommand{\oh}{\ensuremath {\mathit{\overline{h}}}{\xspace}}
\newcommand{\os}{\ensuremath {\mathit{\overline{s}}}{\xspace}}
\newcommand{\tss}{\ensuremath {\mathit{\widetilde{s}}}{\xspace}}
\newcommand{\thh}{\ensuremath {\mathit{\widetilde{h}}}{\xspace}}

\newcommand{\txx}{\ensuremath {\mathit{\widetilde{x}}}{\xspace}}
\newcommand{\tsigma}{\ensuremath {\mathit{\widetilde{\sigma}}}{\xspace}}

\newcommand{\advuls}{\ensuremath {\mathit{Adv}_{\texttt{ESEM}}^{\EUCMA}(t,q_h,\qs)}{\xspace}}

\newcommand{\eat}[1]{}                

\newcommand{\zo}{{\{0,1\}}}

\newcommand{\m}{\ensuremath {m}{\xspace}}

\newcounter{linecounter}

\usepackage[hang,flushmargin]{footmisc}

\usepackage{pifont}

\newtheorem{theorem}{Theorem}
\newtheorem{lemma}{Lemma}

\newtheorem{assumption}{Assumption}

\newcommand\blfootnote[1]{%
	\begingroup
	\renewcommand\thefootnote{}\footnote{#1}%
	\addtocounter{footnote}{-1}%
	\endgroup
}

\begin{document}


\title{Energy-Aware Digital Signatures for \\  Embedded Medical Devices}




\author{
	\IEEEauthorblockN{Muslum Ozgur Ozmen} 
	\IEEEauthorblockA{University of South Florida\\
		Tampa, Florida, USA \\
		ozmen@mail.usf.edu}
	\and
	\IEEEauthorblockN{Attila A. Yavuz}
	\IEEEauthorblockA{University of South Florida\\
		Tampa, Florida, USA \\
		attilaayavuz@usf.edu}
	\and
	\IEEEauthorblockN{Rouzbeh Behnia}
	\IEEEauthorblockA{University of South Florida\\
		Tampa, Florida, USA \\
		behnia@mail.usf.edu}}

\maketitle

\begin{abstract}


Authentication is vital for the Internet of Things (IoT) applications involving sensitive data (e.g., medical and financial systems).  Digital signatures offer scalable authentication with non-repudiation and public verifiability, which are necessary for auditing and dispute resolution in such IoT applications. However, digital signatures have been shown to be highly costly for low-end IoT devices, especially when embedded devices (e.g., medical implants) must operate without a battery replacement for a long time. 

We propose an Energy-aware Signature for Embedded Medical devices (\uls) that achieves near-optimal signer efficiency. \uls~signature generation does {\em not} require any costly operations (e.g., elliptic curve (EC) scalar multiplication/addition), but only a small constant-number of pseudo-random function calls, additions, and a single modular multiplication. \uls~has the smallest signature size among its EC-based counterparts with an identical private key size. We  achieve this by eliminating the use of the ephemeral public key (i.e, commitment) in Schnorr-type signatures from the signing via a distributed construction at the verifier without interaction with the signer while permitting a constant-size public key. We proved that \uls~is secure (in random oracle model), and fully implemented it on an 8-bit AVR microcontroller that is commonly used in medical devices. Our experiments showed that \uls~achieves 8.4$\times$ higher energy efficiency over its closest counterpart while offering a smaller signature and code size. Hence, \uls~can be suitable for  deployment on resource limited embedded devices in IoT. We open-sourced our software for public testing and wide-adoption.
\blfootnote{$\copyright$ 2019 IEEE. Personal use of this material is permitted. Permission from IEEE must be obtained for all other uses, in any current or future media, including reprinting/republishing this material for advertising or promotional purposes, creating new collective works, for resale or redistribution to servers or lists, or reuse of any copyrighted component of this work in other works.}

\end{abstract}



\section{Introduction}\label{sec:Introduction}%
\begin{table*}[t!]
	\centering
	\caption{Signature generation performance of \uls~and its counterparts on AVR ATmega 2560 microcontroller} \label{tab:AVR}
	\vspace{-2mm}
	\begin{threeparttable}
		\begin{tabular}{| c || c | c | c |  c | c | c | }
			\hline
			\textbf{Scheme} & \textbf{CPU cycles} & \textbf{Signing Speed (ms)} & \textbf{Code Size (Byte)} & \textbf{Signature Size (Byte)} & \textbf{Private Key (Byte)} &  \textbf{CPU energy (mJ)}  \\ \hline \hline
			
			%
			%
			%
			%
			%
			%
			
			ECDSA & 79 185 664 & 4949 & 11 990 & 64  & 32 & 494.91 \\ \hline
			
			BPV-ECDSA & 23 519 232 & 1470 & 27 912 & 64 & 10 272 & 146.99 \\ \hline
			
			Ed25519 & 34 342 230 & 2146 & 17 373 & 64 & 32 & 214.64 \\ \hline 
			
			SchnorrQ & 5 174 800 &  323 & 29 894 & 64 & 32 & 32.34 \\ \hline \hline 
			
			\uls & \textbf {616 896} & \textbf {38} & 18 465 & \textbf {48} & 32 & \textbf {3.85} \\ \hline
		\end{tabular}
		\begin{tablenotes}[flushleft] \scriptsize{

				We use low-are implementations due to the memory constraints of ATmega 2560. Note that \uls~does not store any precomputed components (e.g., \BPV~tables)

			}
			
		\end{tablenotes}
	\end{threeparttable}
	\vspace{-5mm}
\end{table*}


It is essential to provide authentication and integrity services for the emerging Internet of Things (IoT) systems that include resource-constrained devices. Due to their computational efficiency, symmetric key primitives (e.g., message authentication codes) are usually preferred for such systems. On the other hand, these primitives might not be scalable for large and ubiquitous systems, and they also do not offer public verifiability and non-repudiation properties, which are essential for some IoT applications~\cite{MedicalDevice:SoK:SP:2014:Rushanan,Ozmen_IOT_SP,MedicalDevice:Survey:2015:Camara2015272}. For instance, in financial IoT applications and implantable medical devices, digital forensics (e.g., legal cases) need non-repudiation and public verifiability~\cite{Ozmen_IOT_SP,MedicalDevice:Survey:2015:Camara2015272,Repudiation:Survey}. Moreover, such systems may include many devices that require scalability. 

Digital signatures rely on public key infrastructures and offer scalable authentication with non-repudiation and public verifiability. Therefore, they are ideal authentication tools for the security of IoT applications. On the other hand, most of the compact digital signatures (e.g., elliptic curve (EC) based signatures) require costly operations such as EC scalar multiplication and addition during   signature generation. It has been shown~\cite{FeasibilityIoT_Crypto,Yavuz:2013:EET:2462096.2462108,BPV:Ateniese:Journal:ACMTransEmbeddedSys:2017}, and further demonstrated by our experiments that, these operations can be energy costly, and therefore, can negatively impact the battery life of highly resource-limited embedded devices. For instance, as one of the many potential applications, we can refer to a resource-limited sensor (e.g., a medical device~\cite{MedicalDevice:SoK:SP:2014:Rushanan}) that frequently generates and signs   sensitive data (medical readings), which are verified by a resourceful cloud service provider. There is a need for lightweight signatures that can meet the computation, memory and battery limitations of these IoT applications.

{\em The goal of this paper is to devise an energy-aware and compact digital signature scheme that can meet some of the stringent battery and memory requirements of highly resource-limited IoTs (e.g., implantable medical devices) that must operate for long periods of time with minimal intervention.} \\ \vspace{-3mm}


\noindent \textbf{Design Objectives}: {\em (i)} The signature generation should not require any costly operation (e.g., exponentiation, EC operations), but only symmetric cryptographic functions (e.g., pseudorandom functions) and basic arithmetics (e.g., modular addition)  {\em (ii)} The low-end devices are generally not only computation/battery but also memory limited. Hence, the objective (i) should be achieved without consorting precomputed storage (e.g., Boyko-Peinado-Venkatesan (\BPV) tables~\cite{BPV:basepaper:1998}, or online/offline signatures~\cite{OfflineOnline_ImprovedShamir_2001}). {\em (iii)} The signing should not draw new randomness~\cite{HighSpeedSignature:Bernstein:Journal2012} to avoid potential hurdles of weak pseudo-random number generators. {\em (iv)} The size of the signature should be small and constant-size as in Schnorr-like signatures. {\em (v)} The size of the public key should be constant.

\noindent \textbf{Our Contributions}: (i) We create an {\em Energy-aware Signature for Embedded Medical devices} (\uls), which is ideal for the signature generation on highly resource-limited IoT devices. We observe that the realizations of Schnorr-like signatures on efficient elliptic curves (e.g., FourQ~\cite{FourQ}) are currently the most efficient solutions, and a generation of the commitment value via a scalar multiplication is the main performance bottleneck in these schemes. Our main idea is to completely eliminate the generation, storage, and transmission of this commitment from the signing of Schnorr signature. To achieve this, we first develop a new algorithm that we call as {\em Signer NOn-interactive Distributed BPV} (\node), which permits a distributed construction of the commitment for a given signature at verifier's side, without requiring any interaction with the signer.  We then transform the signature generation process such that the correctness and provable security are preserved once the commitment value is separated from message hashing and \node~is incorporated into \uls. We present our proposed algorithms in Section \ref{sec:duls_proposed}. In Section \ref{sec:security}, we prove that \uls~is secure in the random oracle model~\cite{RandomOracleModel93} under a semi-honest distributed setting for \node.

(ii) We implemented \uls~and its counterparts both on an AVR ATmega 2560 microcontroller and a commodity hardware, and provided a detailed comparison in Section \ref{sec:PerformanceAnalysis}. We also conducted experiments to assess the battery consumption of \uls~and its counterparts when they are used with common IoT sensors (e.g., a pulse and pressure sensor). We make our implementation open-source for broad testing and adoption. \\ \vspace{-3mm}

\noindent \textbf{Desirable Properties of \uls}: We summarize the desirable properties of our scheme as follows (Table \ref{tab:AVR} gives a comparison of \uls~with its counterparts in terms of signing efficiency on 8-bit AVR processor):

$\bullet$~{\em \underline{Signing and Energy Efficiency}}:  The signature generation of \uls~does not require any EC operations (e.g., scalar multiplication, addition) or exponentiation, but only pseudo-random function (PRF) calls, modular additions and a single modular multiplication. Therefore, \uls~achieves the lowest energy consumption among their counterparts. For example,  \uls~consumes $8$$\times$ and $55$$\times$ less battery than SchnorrQ~\cite{FourQ}, and Ed25519~\cite{Ed25519}, respectively. Our experiments indicate that \uls~can substantially extend the battery life of low-end devices integrated with IoT applications (see Section \ref{sec:PerformanceAnalysis}). Similarly, \uls~is at least a magnitude of times faster than Ed25519 both in an 8-bit microcontroller and commodity hardware. This gap further  increases when our high-speed variant \ulsv~(introduced in Section~\ref{sec:duls_proposed}) is considered.

$\bullet$~{\em \underline{Small Private Key and Signature Sizes:}}  \uls~has the smallest signature size among its counterparts ($48$ Bytes for $\kappa = 128$) with an identical private key size. \uls~does not require any precomputation tables to be stored at the signer, and therefore it is significantly more storage and computation efficient than schemes relying on \BPV~at the signer's side. Moreover, \uls~has a small code size at the signer since it only requires symmetric primitives and basic arithmetics.

$\bullet$~{\em \underline{High Security:}} (i) Side-channel attacks exploiting the EC scalar multiplication implementations in ECDSA were proposed~\cite{DSAExponentiationCCS16}. Since \uls~does not require any EC operations at the signer, it is  not vulnerable to these types of attacks. (ii) The security of Schnorr-like signatures are sensitive to weak random number generators. The signing of \uls~does not consume new randomness (as in \cite{HighSpeedSignature:Bernstein:Journal2012}), and therefore can avoid these problems. (iii) We prove that \uls~is \EUCMA~secure in the random oracle model~\cite{RandomOracleModel93}. \\ \vspace{-3mm}

\noindent \textbf{Potential Use-cases}: In many IoT applications, extending the battery life of low-end processors (i.e., usually signers) is a priority, while verifiers generally use a commodity hardware (e.g., a server) with reasonable storage and communication capabilities. In particular, energy efficiency is a vital concern for embedded medical devices, as they are expected to operate reliably for long periods of time. Currently, symmetric cryptography is preferred to provide security for such devices~\cite{MedicalDevice:General:2011}. At the same time, the ability to produce publicly verifiable authentication tags with non-repudiation is desirable for medical systems~\cite{Ozmen_IOT_SP,MedicalDevice:Survey:2015:Camara2015272,Repudiation:Survey} (e.g., digital forensics and legal cases). Moreover, scalable integration of various medical apparatus to IoT realm will receive a significant benefit from the ability to deploy digital signatures on these devices~\cite{MedicalDevice:SoK:SP:2014:Rushanan}. \uls~takes a step towards meeting this need, as it is currently the most energy efficient alternative with small signature and private key sizes. Essentially, any IoT application involving energy/resource limited signers and more capable verifiers (e.g., wireless sensor networks and IoT sensors in smart cities) are expected to receive benefit from \uls. \\ \vspace{-3mm}

\noindent \textbf{Limitations}: The signature verification of \uls~is distributed, wherein a verifier reconstructs the commitment value of a signature with $l$ parties. Therefore, verification of \uls~is not real-time, and the verifier should wait for a response from all parties. However, as confirmed with our experiments, this only results in a few milliseconds of delay. Moreover, the signer does not need interaction with any parties to compute signatures. Parties aiding the verification are assumed to be semi-honest (do not deviate from the protocol, but try to learn information) and non-colluding (as in traditional semi-honest secure multi-party computation). In our case, even $(l-1)$ parties collude, \uls~remains EU-CMA secure. Since \uls~is designed for a near-optimal signer performance, we believe that \uls~is suitable for applications as outlined above, where a small delay and interaction can be tolerated at the verifier.

\section{Preliminaries and Models}\label{sec:prelim}%
We first give the notations and definitions used by our schemes, and then describe our system/security model. 
\subsection{Notation and Definitions} \label{subsec:Notation}

\noindent \textbf{Notation}:   $||$ and $|x|$ denote   concatenation and the bit length of variable $x$, respectively.  $x\stackrel{\$}{\leftarrow}\mathcal{S}$ means variable $x$ is randomly     selected from set $\mathcal{S}$. $|\mathcal{S}|$ denotes the cardinality of set $\mathcal{S}$. We denote by $\{0,1\}^{*}$ the set of binary strings of any finite length. The set of items $q_i$ for $i=0,\ldots,n-1$ is denoted by $\{q_i\}_{i=0}^{n-1}$. $\log{x}$ denotes $\log_{2}{x}$. $\mathcal{A}^{\mathcal{O}_0,\ldots,\mathcal{O}_{i}}(.)$ denotes
algorithm $\mathcal{A}$ is provided with oracles
$\mathcal{O}_0,\ldots,\mathcal{O}_{i}$. For example,
$\mathcal{A}^{\sgnsig_{sk}}(.)$ denotes algorithm
$\mathcal{A}$ is provided with a {\em signing oracle} of
algorithm $\mathit{Sig}$ of signature scheme $\sgn$ under a private key $sk$. We define a pseudo-random function (\texttt{PRF}) and three hash functions to be used in our schemes as follows:  $\ensuremath{PRF}_{0}:\zo^{*} \rightarrow \zo^{\kappa}$, $\ensuremath{H}_{0}:\zo^{*} \rightarrow \zo^{\kappa}$, $\ensuremath{H}_{1}:\zo^{*} \rightarrow \zo^{v\cdot \log{n}}$ and $\ensuremath{H}_{2}:\zo^{*}\rightarrow \mathbb{Z}_{q}^{*}$, where $1<v<n$ are BPV parameters and $\kappa$ is the security parameter.

\begin{definition} \label{Def:GenericSig}
	A signature scheme \sgn~is a tuple of three algorithms
	$(\mathit{Kg},\mathit{Sig},\mathit{Ver})$ defined as follows:
	\begin{enumerate}[-]
		\itemsep0.2em
		\item $\underline{(\sk,\pk)\leftarrow \sgnkg(1^{\kappa})}$: Given the security parameter $1^{\kappa}$, the key generation algorithm
		returns a private/public key pair $(\sk,\pk)$.
		\item $\underline{\sigma\leftarrow \sgnsig(\m,\sk)}$: The signing algorithm takes a message $\m$~and a \sk, and returns a
		signature $\sigma$.
		\item $\underline{b\leftarrow \sgnver(\m,\sigma,\pk)}$: The verification
		algorithm takes a message \m, signature $\sigma$ and the public key \pk~as input. It returns a bit $b$: $1$  means {\em valid} and $0$ means {\em invalid}.
	\end{enumerate}
\end{definition}

Our schemes are based on \sch~signature~\cite{Schnorr91}. 	
\begin{definition} \label{Def:Schnorr}
	\sch~signature scheme is a tuple of three algorithms
	$(\mathit{Kg},\mathit{Sig},\mathit{Ver})$ defined as follows:
	\begin{enumerate}[-]
		\itemsep0.2em
		\item $\underline{(y,Y) \leftarrow \schkg(1^{\kappa})}$: Given $1^{\kappa}$ as the input, 
		\begin{enumerate}[1)]
			\item The system-wide $\params\as(q,p,\alpha)$, where $q$ and $p$ are large primes such that $p>q$ and $q |(p-1)$, and a generator $\alpha$ of the subgroup $G$ of order $q$ in $\mathbb{Z}_{p}^{*}$.  
			\item  Generate private/public key pair $(y\stackrel{\$}{\leftarrow}\mathbb{Z}_{q}^{*},Y\leftarrow
			\alpha^{y} \bmod p)$. We suppress \params~afterwards for the brevity. 
		\end{enumerate}
		
		\item $\underline{(\sigma \leftarrow \schsig(m,y)}$: Given $m\in\{0,1\}^{*}$ and $y$ as the input, it returns a
		signature $\sigma=(s,e)$, where $H:\{0,1\}^{*} \rightarrow \mathbb{Z}_{q}^{*}$ is a full domain hash function.
		\begin{enumerate}[1)]
			\item $r\Rq,~~R\as \alpha^{r} \bmod p$.
			\item $e\as H(m||R),~s\as (r-e\cdot y) \bmod q$.
		\end{enumerate}
		
		\item $\underline{b\leftarrow \schver(m,\langle s,e \rangle,Y)}$: The signature verification algorithm takes
		$m$, $\langle s,e \rangle$ and $Y$ as the input.  It computes $R'\as Y^{e}\alpha^{s} \bmod p$ and returns a bit $b$, with $b=1$
		indicating {\em valid}, if $e=H(m||R')$ and $b=0$
		otherwise.
	\end{enumerate}
\end{definition}

We use Boyko-Peinado-Venkatesan (\BPV)~generator~\cite{BPV:basepaper:1998}. 

\begin{definition} \label{Def:BPV}
	The \BPV~generator is a tuple of two algorithms
	$(\mathit{Offline},\mathit{Online})$ defined as follows:
	\begin{enumerate}[-]
		\itemsep0.2em
		\item $\underline{(\Gamma,v,n,q,p) \as \BPVOff(1^{\kappa})}$: The offline \BPV~algorithm  takes $1^{\kappa}$ as
		the input and generate system-wide parameters $(q,p,\alpha)$ as in $\schkg(1^{\kappa})$.
		\begin{enumerate}[1)]
			\item \BPV~parameters $n$ and $v$ are the number of pairs to be precomputed and the number of elements to be randomly selected out $n$ pairs, respectively, for $2<v<n$.
			\item $r_i\Rq,$~~ $R_i \as \alpha^{r_i} \bmod p$, $i=0,\ldots, n-1$.
			\item Set precomputation table $\Gamma=\{r_i,R_i\}_{i=0}^{n-1}$.
		\end{enumerate}
		
		\item $\underline{(r,R) \as \BPVOn(\Gamma,v,n,q)}$: The online \BPV~algorithm takes the table $\Gamma$ and $(v,n,q)$ as input.
		\begin{enumerate}[1)]
			\item Generate a random set $S \subset [0,n-1]$, where $|S| = v$.
			\item $r \as \sum_{i \in S}^{} r_i \bmod q$, $R \as \prod_{i \in S}^{} R_i$.
		\end{enumerate}
		
	\end{enumerate}
\end{definition}

\begin{lemma} \label{lem:BPVOutputRandom}
	The distribution of \BPV output $r$ is {\em statistically close} to the uniform random distribution with an appropriate choice
	of parameters $(v,n)$ \cite{BPV:basepaper:1998}.
\end{lemma}

\subsection{System  and Security Model} \label{subsec:SecurityModel}

As depicted in Figure~\ref{fig:system}, our system model includes a highly resource-limited signer that computes signatures to be verified by any receiver. Our system model also includes  $l$ distinct parties ($P_1, \ldots, P_l$) that are involved in signature verification. In the line of~\cite{GoldbergPIR}, after the initialization phase, we consider a synchronous network which consists of a client (verifier in \uls) and semi-honest servers $ P= (P_1, \ldots, P_l) $. We assume that the communication channels are secure.

\begin{figure}[t!]
	\vspace{0mm}
	\centering

	\includegraphics[width=\linewidth,height=2.5cm]{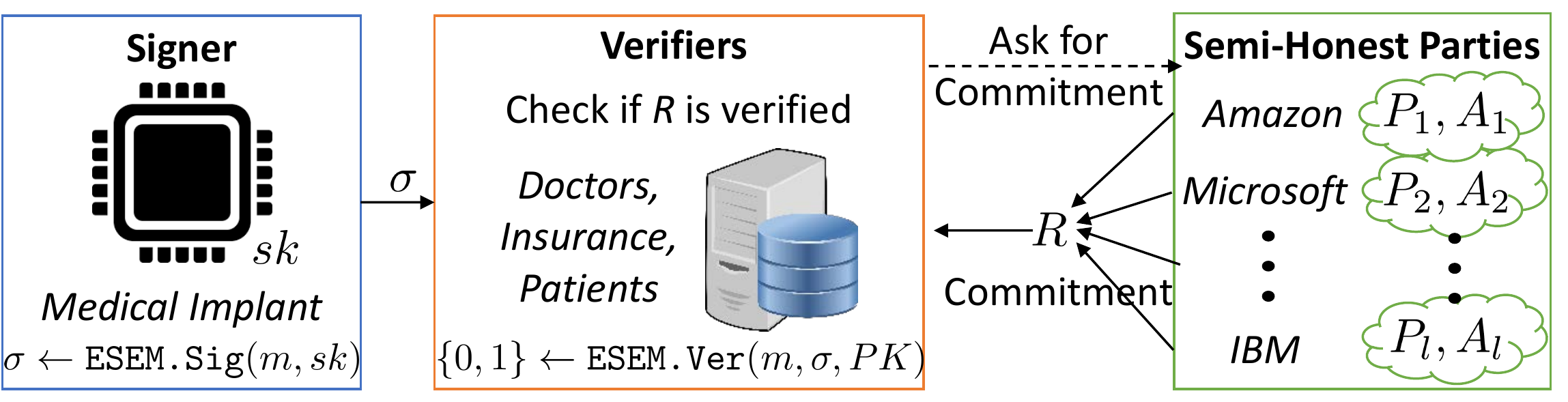}

	\caption{System Model of \uls} \label{fig:system}
\end{figure}

The security notion for a digital signature is Existential Unforgeability against Chosen Message Attack (\EUCMA)~\cite{JonathanKatzModernCrytoBook}. 

\begin{definition}\label{def:EUCMA}
	EU-CMA experiment $ \mathit{Expt}^{\textit{EU-CMA}}_{\mathtt{SGN}} $   for  a signature scheme $  \texttt{SGN} = (\texttt{Kg,Sig,Ver}) $   is defined as follows.
	\begin{itemize}\normalfont
		\item  $ (\sk,pk)\leftarrow \texttt{SGN.Kg}(1^\kappa) $
		\item $ (M^*, \sigma^*)\leftarrow \mathcal{A}^{\texttt{SGN.Sig}(\cdot)} (pk) $
		
	\end{itemize}
	
	\A wins the above experiment if  $  1 \leftarrow {\texttt{SGN.Ver}(M^*, \sigma^*, pk)} $ and $ m^* $ was not queried to  $ \texttt{SGN.Sig}(\cdot) $ oracle.  The EMU-CMA advantage $ \mathit{Adv}^{\textit{EU-CMA}}_{\texttt{SGN}} $  of  \A is defined as   $\Pr[ \mathit{Expt}^{\textit{EU-CMA}}_{\texttt{SGN}}= 1]$
\end{definition}

\begin{definition} \label{Def:l-private}
	A protocol is $ t $-private \cite{GoldbergPIR}  if any set of parties $\mathcal{S}$ with $ |\mathcal{S}| \leq t $ are not able to compute or achieve any output or knowledge any different than what they could have done individually from their set of private input and outputs. 
	
\end{definition}

 \begin{assumption}\label{assump:colluding}
 	We assume that the servers are \emph{semi-honest} - always follow the protocol, but try to learn as much as possible from the shared or observed information.
 
 \end{assumption}


For $ t=l-1 $, where $ l  $ is the total number of the servers, our proposed scheme is $ t $-private. The signature generation in our scheme does not require the participation of the servers. In other words, the signer does not need to interact with any of the $  l$ servers during the signature generation. The participation of all $ l $ servers is however required on the verifier's side. 

\section{Proposed Schemes}\label{sec:duls_proposed}%
We first discuss the design challenges to achieve our objectives outlined in Section \ref{sec:Introduction}. We elaborate our Signer NOn-interactive Distributed BPV (\node) algorithm that addresses some of them. We then present \uls~that uses \node~and other strategies to achieve our objectives. 

\subsection{High-Level Design} \label{subsec:DesignChallange}
Schnorr-like signatures with implementations on recent ECs (e.g., FourQ~\cite{FourQ}) are currently among   the most efficient and compact digital signatures. Hence, we take them as our starting point. In these schemes, the signer generates a random value $r$ and its commitment $R = \alpha^r \bmod p$, which is incorporated into both signing and verification (as an input to hash along with a message). This exponentiation (EC scalar multiplication) constitutes the main cost of the signature generation, and therefore we aim to {\em completely} eliminate it from the signing. However, this is a highly challenging task. 

\subsubsection{Commitment Generation without Signer Interaction} The elimination of commitment $R$ from the signing permits removal of EC operations such as scalar multiplication/additions. It also eliminates the transmission of $R$ and a storage of \BPV~table at the signer. However, the commitment is necessary for the signature verification. Hence, the verifier should obtain a correct commitment for each signature with the following requirements: (i) The verifier cannot interact with the signer to obtain the commitment (i.e., the signer does not have it). (ii) The signer non-interactive construction of the commitment should not reveal the ephemeral randomness $r$. (iii) Unlike some multiple-time signatures~\cite{Yavuz:2013:EET:2462096.2462108}, the verifier should not have a fixed limit on the number of signature verifications and/or a linear-size public key. 

We propose a new algorithm that we refer to as \node,~to achieve these requirements. Our idea is to create a distributed BPV technique that permits a set of parties to construct a commitment on behalf of the signer. This distributed scheme permits the verifiers to obtain the corresponding commitment of a signature from these parties on demand without revealing $r$ or an interaction with the signer. We elaborate on \node~in Section \ref{subsec:NODE}.

\subsubsection{Separation of the Commitment from Signature Generation with \node} The commitment value is generally used as a part of message hashing (e.g., $H(M||R)$ in \sch) in Schnorr-like signatures. To eliminate $R$ from the signature, the commitment must be separated from the message. However, the use of commitment in the message hashing plays a role in the security analysis of Schnorr-like signatures. Moreover, the removal of commitment $R$ from the signing while using $r$ with \node~algorithm requires a design adjustment.

We propose our main scheme \uls~that achieves these goals. In the line of~\cite{Yavuz:2013:EET:2462096.2462108}, we use a one-time random value $x$ in the message hashing, but also devise an index encoding and aggregate BPV approach to integrate \node~into signature generation. This permits a constant-size public key at the verifier without any interaction with the signer. We give the details of \uls~in Algorithm \ref{alg:ULS_Generic}. 

\subsection{Signer NOn-interactive Distributed BPV (\node)} \label{subsec:NODE}
We conceive \node~as a distributed realization of \BPV~\cite{BPV:basepaper:1998}~where $l$ parties hold $n$ public values $\{R_{i,j}\}_{i=1,j=1}^{n,l}$ of $l$ \BPV~tables, and then can collaboratively derive $R$ without learning its corresponding private key $r$ unless all of them collude. We stress that one cannot simply shift the storage of public values in a \BPV~table to a single verifier. This is because the indexes needed to compute the commitment $R$ should remain hidden in order to protect the one-time randomness $r$~\cite{BPV:DistributionofmodularSum:2001:Nguyen}. We overcome this challenge by creating a distributed \BPV~approach that can be integrated into a Schnorr-like signature. At $\nodekg$, in Step 2, the secret key $y$ is used as a seed to derive $l$ secret values $\{z_j\}_{j=1}^{l}$. Each $z_j$ is used to deterministically generate secret \BPV~values $\{r_{i,j}\}_{i=1,j=1}^{n,l}$, whose corresponding public values $\{R_{i,j}\}_{i-1,j=1}^{n,l}$ are computed in Step 4-5 and given to parties $(P_1,\ldots,P_l)$.

At the online phase, the sender (i.e., signer) generates the aggregated $r$ on its own and the receiver (i.e., verifier) generates the aggregated $R$ cooperatively with the parties $(P_1,\ldots,P_l)$. The sender first derives a random value $x$ from a keyed hash function (at $\nodesig$ Step 1), and then deterministically derives $z_j$ values (Step 3) as in $\nodekg$. Sender uses the $z_j$, that is only shared with the corresponding party, and the one-time random value $x$ to generate the set (indexes) to be used to aggregate the values. This step is of high importance since this way, the sender {\em commits} to the one-time random value $x$. Sender repeats this process for all $l$ parties and aggregates (adds) all the corresponding $r_{i,j}$s to derive the resulting $r$ (Step 5). 

The verifier proceeds as follows to generate the corresponding $R = \alpha^r \bmod p$.  At Step 1 in $\nodever$, the verifier communicates with $l$ parties to derive each $R_j$ from them. Upon request, parties first derive the same set (indexes) as the sender (Step 1 in $\texttt{SNOD.P}_j\mhyphen\texttt{Construct}$). Then, each party aggregates the corresponding $R_{i,j}$ that were assigned to them in $\nodekg$, and returns the results to the verifier. The verifier aggregates all these values at Step 2, to derive the corresponding $R$. Please note that only the parties can create the set (indexes) since only they have their corresponding $z_j$ values. Moreover, since all servers provide $R_j$ that can be generated only by them, unless all of the servers collude, they cannot learn any information about the other indexes or the one-time randomness $r$. This makes our scheme $t$-private, as shown in Lemma~\ref{lem:l-private}.

\begin{algorithm}[t!]
	\caption{Signer NOn-interactive Distributed BPV (\node)}\label{alg:NODE}
	\hspace{5pt}
	\begin{algorithmic}[1]
		
		\Statex   $\underline{(A_1, \ldots, A_l)\as \nodekg(1^{\kappa},y,v,n)}$: Given $1^{\kappa}$, secret key $y$, and parameters $(v,n)$ generate precomputation tables.
		\For{$j = 1, \ldots, l$}
		\State $z_j \as \texttt{PRF}_0(y||j)$
		\For{$i = 1, \ldots, n$}
		\State $r_{i,j} \as \texttt{PRF}_0(z_j||i)$
		\State $R_{i,j} \as \alpha^{r_{i,j}} \bmod p$
		\EndFor 
		\State Set $A_j = (z_j, v,\langle R_{1,j}, \ldots ,R_{n,j} \rangle )$
		\EndFor
		\State \Return each $A_j$ to corresponding party $P_j$

	\end{algorithmic}
	\algrule
	
	\begin{algorithmic}[1]
		\Statex $\underline{(r,x) \as \nodesig(sk)}$: 
		\State $x \as H_0(sk||c)$, $c \as c+1$
		\For{$j = 1, \ldots, l$}
		\State $z_j \as \texttt{PRF}_0(sk||j)$
		\State $(i_{1,j}, \ldots, i_{v,j}) \as H_1(z_j||x)$
		\EndFor 
		\State $r \as \sum_{k = 1}^{v}\sum_{j=1}^{l}\texttt{PRF}_0(z_j||i_{k,j}) \bmod q$
		\State \Return $r$
	\end{algorithmic}
	\algrule
	
	\begin{algorithmic}[1]
		
		\Statex   $\underline{\bar{R}_j \as \texttt{SNOD.P}_j\mhyphen\texttt{Construct}(A_j,x)}$: Given $x$ and $A_j$, each party $P_j$ returns $\bar{R}_j$.
		\vspace{3pt}
		\State $(i_{1,j}, \ldots, i_{v,j}) \as H_1(z_j||x)$
		\State $\bar{R}_j \as \prod_{k=1}^{v}R_{i_{k,j},j} \bmod p$
	\end{algorithmic}
	
	\algrule
	
	\begin{algorithmic}[1]
		\Statex $\underline{R \as \nodever(x)}$ Given $x$, retrieve its commitment $R$  under \pk~from $(P_1,\ldots,P_l)$.
		\State Verifier sends $x$ to the parties, and each party $P_j$ returns $\bar{R}_j \as \texttt{SNOD.P}_j\mhyphen\texttt{Construct}(A_j,x)$ for $j = 1, \ldots, l$.
		\State $R \as \prod_{j=1}^{l}\bar{R}_j \bmod p$
		\State \Return $R$
	\end{algorithmic}
\end{algorithm}

\subsection{Energy-aware Signature for Embedded Medical devices} \label{subsec:ULS}
We summarize our main scheme \uls~(see Algorithm ~\ref{alg:ULS_Generic}), which permits a near-optimal signing by integrating \node~into \texttt{Schnorr} signature with alternations.

During key generation, secret/public key pair ($y,Y$) and \BPV~parameters are generated (Step 1-2), followed by $\nodekg$\any~algorithm to  obtain the distributed \BPV~public values to be stored by parties $(P_1,\ldots,P_l)$. In $\ulssig$\any, the signer generates the ephemeral random value $r$ and one-time randomness $x$ to be used as the commitment. Instead of the commitment in \texttt{Schnorr} ($R$), the signer uses $x$ as the commitment in Step 2. This separation of the commitment $R$ from the message hashing is inspired from~\cite{Yavuz:2013:EET:2462096.2462108}. Note that, unlike the multiple-time signature in~\cite{Yavuz:2013:EET:2462096.2462108} that can only compute a constant pre-determined number of signatures with a very large linear-size public key, \uls~can compute polynomially unbounded number of signatures with a constant public key size.  Finally, the verifier first calls the $\nodever$\any~algorithm to generate the public value $R$, by collaborating with the parties. The signature verification, which is similar to \texttt{Schnorr} with the exception of the commitment $x$, is performed at Step 2.

\begin{algorithm}[h!]
	\caption{Energy-aware Signature for Embedded Medical devices (\uls)}\label{alg:ULS_Generic}
	\hspace{5pt}
	\begin{algorithmic}[1]
		
		\Statex   $\underline{(\sk,\pk, A_1, \ldots, A_l)\as \ulskg(1^{\kappa})}$: 
		\vspace{3pt}
		\State $(y,Y) \leftarrow \schkg(1^{\kappa})$, where $(q,p,\alpha)$ as in \schkg.
		\State Select $(v,n)$ such that $\binom{n}{v} \geq 2^{\kappa}$
		\State $(A_1, \ldots, A_l) \as \nodekg({1^{\kappa}, y,v,n})$
		\State The signer stores $ \sk=(y)$ and  $(v,n,q, c = 0)$. The verifier stores $\pk=Y$ and $(v,n,q,p,\alpha)$.
	\end{algorithmic}
	\algrule
	
	\begin{algorithmic}[1]
		\Statex $\underline{\sigma\as \ulssig(m,\sk)}$: 
		\vspace{3pt}
		\State $(r,x) \as \nodesig(y)$
		\State $s \as r - H_2(m||x) \cdot y \bmod q$
		\State \Return $\sigma=(s,x)$.
	\end{algorithmic}
	\algrule
	
	\begin{algorithmic}[1]
		\Statex $\underline{\zo\as \ulsver(m,\sigma,\pk)}$: 
		\vspace{3pt}
		\State $R \as \nodever(x)$
		\State \textbf{if} {$R=Y^{H_2(m||x)} \cdot \alpha^{s} \bmod p$}, \Return 1, else \Return 0.
	\end{algorithmic}
\end{algorithm}

\subsubsection{\ulsv} \label{subsec:ULSv}
We point out a trade-off between the private key size and signing speed, which can increase the signature generation performance with the cost of some storage. The signer can store private keys $\{r_{i,j}\}_{i=1,j=1}^{n,l}$ in the memory, and therefore avoid $v \cdot l$ \texttt{PRF} invocations. {\em We refer to this simple variant as \ulsv}. As demonstrated in Section \ref{sec:PerformanceAnalysis}, an extra storage of $12$ KB can boost the performance of \uls~commodity hardware.

\section{Security Analysis} \label{sec:security}

\begin{lemma}\label{lem:l-private}
	The scheme proposed in Algorithm \ref{alg:NODE} is $t$-private (in the sense of Definition \ref{Def:l-private}) with regard to $ r $  and therefore, it can resist against $ l -1 $ colluding servers.  
	
\end{lemma}
\begin{proof}
	The random values   $ \{r_{i,j} \gets  \texttt{PRF}_0(z_j||i) \}_{j=1,i=1}^{l,n}$  are generated uniformly at random  in the \nodekg$ (\cdot) $ via private seed $  z_j $,  which is given to each server  $ P_j $.   The  security of $ \node $ (i.e., \uls) relies on the secrecy of $ r  \gets  \sum_{i = 1}^{l}\sum_{j=1}^{n}r_{i,j} \bmod q$. Given each $ r_{i,j}  $ is generated uniformly at random via $ z_j $'s, and due to Lemma \ref{lem:BPVOutputRandom}, for the adversary  \A   to infer $ r $, it must know all $ l $ private seeds  $ z_j $ or corrupt all of the  $ l $ servers. 
\end{proof}

\begin{theorem} \label{the:ULSSecurityTheorem}
	In the random oracle model, based on Assumption \ref{assump:colluding}  and Lemma \ref{lem:l-private}, if a polynomial-time adversary \A can break the EU-CMA  security of   \uls~in time $ t $ and after $ q_h $ hash and $ q_s  $ signature queries, then one can build polynomial-time algorithm \F that breaks the EU-CMA  security of  \sch~signature in time $ t' $ and $ q_{s }' $ signature queries. 
	\begin{eqnarray*}
		\advuls & \le & \advsch ,
	\end{eqnarray*}
\end{theorem}
\noindent {\em Proof:} Please refer to the Appendix.

\section{Performance Analysis} \label{sec:PerformanceAnalysis}

\subsection{Parameter Selection}

We select FourQ curve~\cite{FourQ} that offers fast elliptic curve operations (that is desirable for our verification process, remark that signer has no EC operations) with $128$-bit security level. The selection of parameters $(v,n)$ relies on the number of $v$-out-of-$n$ different combinations possible. We select $n= 1024$, $v = 18$ for \uls~and $n= 128$, and $v = 40$ for \ulsv, where both offers over $2^{128}$ different combinations. Lastly, we select $l =3$ (i.e., 3 parties are involved in verification).

\begin{figure*}[t!]
	\vspace{0mm}
	\centering
	\begin{subfigure}{.45\textwidth}
		
		\includegraphics[width=\linewidth]{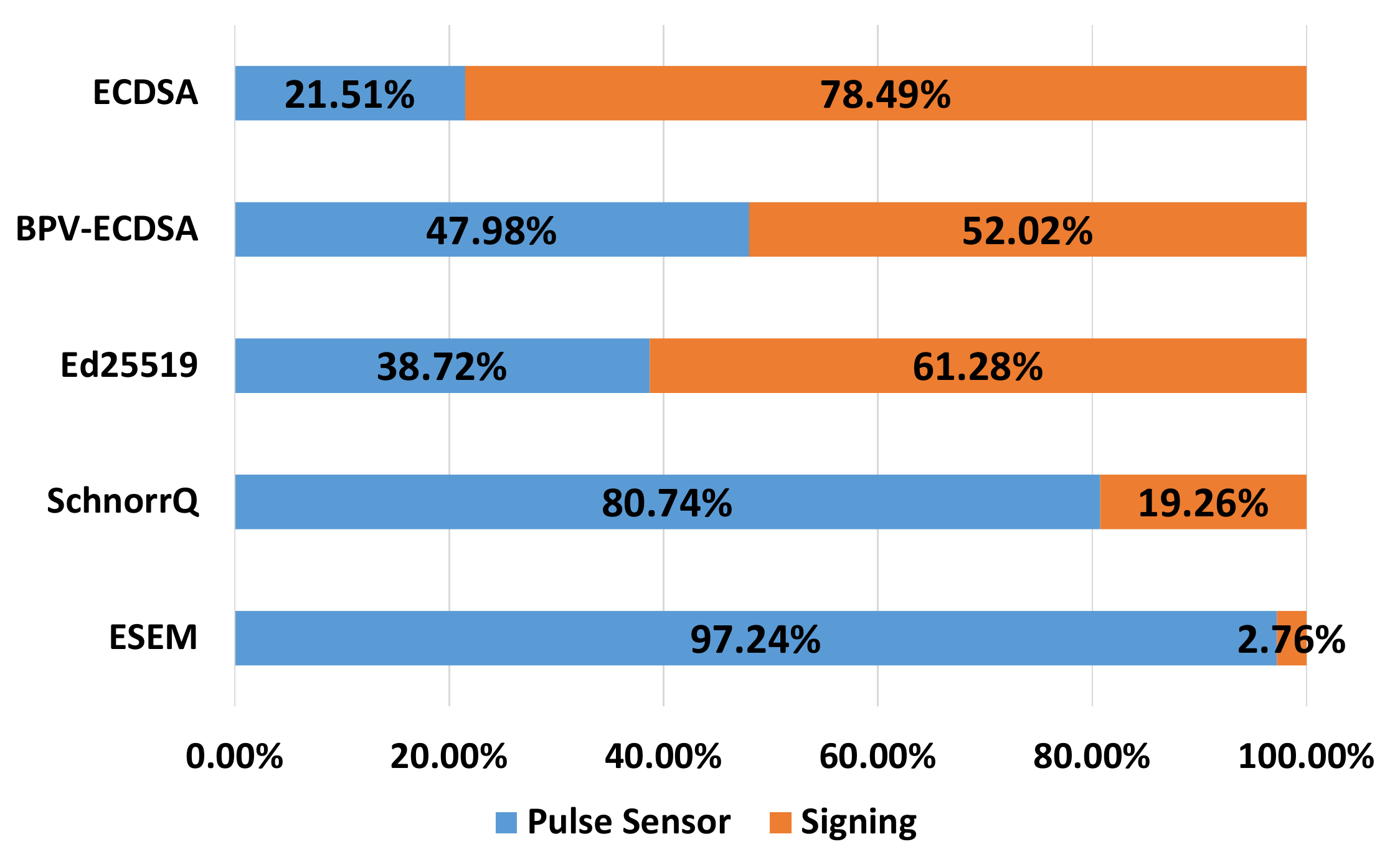}
		\caption{Energy of Signature Generation vs Pulse Sensor}
		\label{fig:areaPulse}
		\vspace{-1mm}
	\end{subfigure}%
	\begin{subfigure}{.45\textwidth}
		
		\includegraphics[width=\linewidth]{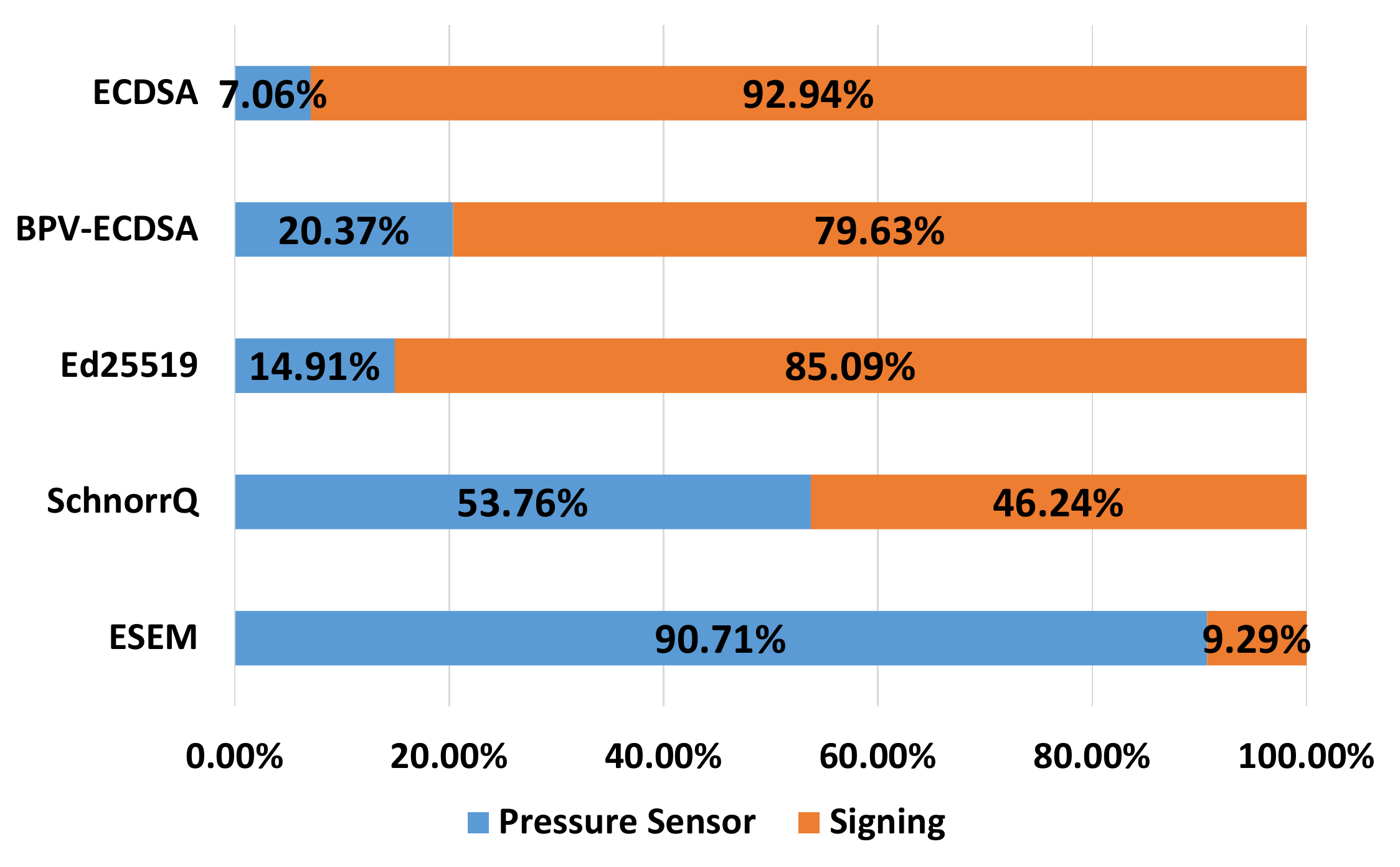}
		\caption{Energy of Signature Generation vs Pressure Sensor}
		\label{fig:areaPressure}
		\vspace{-2mm}
	\end{subfigure}
	\caption{Energy consumption of signature generation vs IoT sensors} \label{fig:areaBattery}
	\vspace{-6mm}
\end{figure*}

\subsection{Evaluation Metrics and Experimental Setup} \label{subsec:ExpSetupMetrics}

\noindent \textbf{Evaluation Metrics}: We implemented \uls~and its counterparts both on the low-end device (8-bit microcontroller) and a commodity hardware. (i) At the signer's side, the signature generation time and private key size were evaluated on both types of devices. The energy consumption and code size were evaluated on a low-end device. (ii) The signature size is evaluated as the communication overhead. (iii) At the verifier's side, the signature verification time and the size of public key were evaluated on the commodity hardware. 

Note that the time required to transmit the \uls~signature (only $48$ Bytes) is  already smaller than all of its counterparts. Therefore, we do not include this in our experiments. The  bandwidth overhead to construct $R$ between the verifier and $l$ parties is only $48$ Bytes, and highly depends on the geographic location of the server (i.e., round trip time). We conservatively benchmark this network delay and include in our signature verification time, with an Amazon EC2 server in North Virginia.

\noindent \textbf{Hardware Configurations and Software Libraries}: We selected AVR ATmega 2560 microcontroller as our low-end device due to its low power consumption and extensive use in practice, especially for medical devices \cite{MedicalDevice:SoK:SP:2014:Rushanan,Ozmen_IOT_SP,ATmega2560Medical}. It is an 8-bit microcontroller with $256$ KB flash memory, $8$ KB SRAM, $4$ KB EEPROM and maximum clock speed is $16$ MHz.

We implemented our schemes using Rhys Weatherley's library\footnote{https://github.com/rweather/arduinolibs/tree/master/libraries/Crypto}, which enables Barrett reduction to compute modulo $q$. We used BLAKE2s~\cite{blakeHash} as our hash function from the same library, since it is optimized for low-end devices in terms of speed and code size. We instantiated our \texttt{PRF} function as CHACHA20 stream cipher~\cite{BernsteinCHACHA} which offers high efficiency. To assess our counterparts, we used ECDSA implementation in microECC\footnote{https://github.com/kmackay/micro-ecc}, with which we also implemented BPV-ECDSA. We used the implementations on same microcontroller to assess Ed25519~\cite{Ed25519AVR} and SchnorrQ~\cite{FourQ8bit}.

We powered the microcontroller with a $2200$ mAh power pack. ATmega 2560 operates at a voltage level of $5$ V and takes $20$ mA current\footnote{http://www.atmel.com/Images/Atmel-2549-8-bit-AVR-Microcontroller-ATmega640-1280-1281-2560-2561\_datasheet.pdf}. We verified the current readings taken from datasheets by connecting an ammeter between the battery and ATmega 2560, and we observed an insignificant difference. Therefore, we measured the energy consumption with the formula $E = V \cdot I \cdot t$ where $t$ is the computation time. To account the variations in time $t$, we run each scheme $10^4$ times and took the average.

We also investigated the effect of cryptography on the battery life in some real-life IoT applications. For this purpose, we measured the energy consumption of a pulse sensor\footnote{https://pulsesensor.com/} and a BMP183 pressure sensor\footnote{https://cdn-shop.adafruit.com/datasheets/1900\_BMP183.pdf}. We expect that the pulse and pressure sensors provide some ideas on the use of digital signatures with sensors in medical devices and daily IoT applications, respectively.

$\bullet$~{\em Commodity Hardware}: We used an Intel i7-6700HQ $2.6$ GHz processor with $12$ GB of RAM as the commodity hardware in our experiments. We implemented the arithmetic and curve operations of our scheme with FourQlib\footnote{https://github.com/Microsoft/FourQlib}. We used BLAKE2b~\cite{blakeHash} as our hash function since it is optimized for commodity hardware. Lastly, we instantiated our \texttt{PRF} with AES in counter mode using Intel intrinsics. For our counterparts, we used their base implementations.

As the semi-honest party, we used an Amazon EC2 instance located in North Virginia. Our EC2 instance was equipped with an Intel Xeon E5 processor that operates at $2.4$ GHz. 

Our implementations are open-sourced at:
\vspace{-1mm}
\begin{center}
	\fbox{\url{www.github.com/ozgurozmen/ESEM}}
\end{center}
\vspace{-2mm}

\subsection{Performance Evaluation and Comparisons} \label{subsec:PerformanceEval}

\noindent \textbf{Low-end Device}: Table~\ref{tab:AVR} shows the results obtained from our implementations on 8-bit AVR ATmega 2560.

$\bullet$~{\em Signature Generation Speed}: \uls~has the fastest signing speed, which is $8.4$$\times$ and $55$$\times$ faster than that of SchnorrQ and Ed25519, respectively.

\begin{table*}[t!]
	\centering
	\caption{Experimental performance comparison of \uls~schemes and their counterparts on commodity hardware} \label{tab:Laptop}
	\vspace{-2mm}
	\begin{threeparttable}
		\begin{tabular}{| c || c | c | c | c |  c | c | c | c | }
			\hline
			\textbf{Scheme} & \specialcell[]{\textbf{Signing}\\  \textbf{ Time (}$\mu$s\textbf{)}} & \specialcell[]{\textbf{Private Key}\textsuperscript{$ \mathparagraph$} \\ \textbf{(Byte)}} & \specialcell[]{\textbf{Signature }\\  \textbf{Size (Byte)}} & \specialcell[]{\textbf{Verifier}\\  \textbf{Comp. (}$\mu$s\textbf{)}} & \specialcell[]{\textbf{Verifier} \\ \textbf{Storage (Byte)}} & \specialcell[]{\textbf{Server}\\  \textbf{Comp. (}$\mu$s\textbf{)}} & \specialcell[]{\textbf{Server} \\ \textbf{Storage (KB)}} & \specialcell[]{\textbf{End-to-End} \\ \textbf{Delay\textsuperscript{$\dagger$} (}$\mu$s\textbf{)}} \\ \hline \hline
			
			
			ECDSA & 725 & 32 & 64 & 927 & 32 & $-$ & $-$ & 1652 \\ \hline
			
			BPV-ECDSA & 149 & 10272 & 64 & 927 & 32 & $-$ & $-$ & 1076 \\ \hline
			
			Ed25519 & 132 & 32 & 64 & 336 & 32 & $-$ & $-$ & 468 \\ \hline
			
			SchnorrQ & 12 & 32 & 64 & 22  & 32 & $-$ & $-$  &  34 \\ \hline \hline
			

			\uls & \textbf {11} & 32 & \textbf {48} & 24 & 32 & 5 & 32784 & $11+24+5+\Delta$ \\ \hline
			
			\ulsv &\textbf {4} & 12416 & \textbf {48} & 24 & 32 & 10 & 4112 & $4+24+10+\Delta$ \\ \hline
		\end{tabular}
		\begin{tablenotes}[flushleft]\scriptsize{  
				$ \mathparagraph $ System wide parameters \params~(e.g., p,q,$\alpha$) for each scheme are included in their corresponding codes, and private key size denote to specific private key size.
				
				$\dagger$ $\Delta$ represents the communication between the verifier and servers. Since the verifier communicates with $l = 3$ servers, the maximum communication delay is included in our end-to-end delay. This communication is measured to be $37$ ms on average by our experiments, with an Amazon EC2 instance in N. Virginia. 
				
				
			}
		\end{tablenotes}
	\end{threeparttable}
	\vspace{-5mm}
\end{table*}

$\bullet$~{\em Energy Consumption of Signature Generation}: With a $2200$ mAh battery, \uls~can generate nearly \em $800000$ \em signatures, whereas SchnorrQ, Ed25519 and ECDSA can generate only $94482$, $14235$ and $6173$ signatures, respectively. This shows that, \uls~can generate significantly higher number of signatures with the same battery.

$\bullet$~{\em Energy Consumption of Signature Generation versus IoT Sensors}: We considered a pulse and a pressure sensor to exemplify the potential medical and home automation IoT applications, respectively. We selected the sampling time (i.e., the frequency of data being read from the sensor) as every 10 seconds and every 10 minutes for the pulse and pressure sensor, respectively, to reflect their corresponding use-cases. We measured the energy consumption by considering three aspects: (i) Each sensor by default draws a certain energy as specified in its datasheet. The pulse sensor operates at $3$ V and draws $4.5$ mA of current, while pressure sensor operates at $2.5$ V and draws 5 $\mu$A of current. These values are multiplied by their corresponding sampling rates to calculate the energy consumption of the sensor. (ii) AVR ATmega 2560 consumes energy to make readings from the sensor as well as during its waiting time. We measured the time that takes the microcontroller to have a reading from the sensor as $1$ ms. Therefore, we calculated the energy consumption of the microcontroller on active time as $5V \cdot 20mA \cdot 1ms$. (iii) ATmega 2560 requires $10$ $\mu$A in power-save mode, which is used to calculate the energy consumption in the idle time.

We compared the energy consumption of signature generation and IoT sensors in Figure~\ref{fig:areaBattery}. {\em \uls~reduces the energy consumption of signature generation to $2.76$\% and $9.29$\% compared to that of pulse and pressure sensors, respectively}. Observe that, compared with the pressure sensor, SchnorrQ as the fastest counterpart of \uls,~requires $46.24$\%, while Ed25519 demands $85.09$\% of the energy consumption. When the pulse sensor is used, while \uls~requires an almost negligible energy consumption ($2.76$\%), its closest counterpart requires $19.29$\%. The energy efficiency of \uls~also translates into longer battery life in these applications. More specifically, when pressure sensor is deployed with \uls, it takes $511$ days to drain a $2200$ mAh battery, while it is $303$ days for our closest counterpart (SchnorrQ).

Our experiments show that the existing ECC-based digital signatures consume more energy than IoT sensors, which make them the primary source of battery consumption. On the other hand, {\em \uls~was able to reduce the signature generation overhead to a potentially negligible level in some cases, at minimum offering improvements over its counterparts.}

\noindent \textbf{Commodity Hardware}: The benchmarks of \uls~and its counterparts on commodity hardware are shown in Table~\ref{tab:Laptop}. 

$\bullet$~{\em Signature Generation}: \uls~and \ulsv~schemes offer the fastest signature generation on commodity hardware as well. Especially \ulsv~(the high-speed variant where private key components are stored instead of generating them from a seed), is 3$\times$ faster than its closest counterpart.

$\bullet$~{\em Signature Verification}: The signature verification in \uls~includes verifier computation, server computation and communication between the verifier and servers. Due to the computational efficiency of FourQ curve, verifier and server computation of \uls~verification is highly efficient. Specifically, verifier computation takes $24$ $\mu$s in \uls~and \ulsv; and server computation takes $5$ $\mu$s, and $10$ $\mu$s for \uls~and \ulsv, respectively. The communication between server and verifier is experimented with our commodity hardware and an Amazon EC2 instance at N. Virginia. This delay was measured as $37$ ms on average.

The fastest verification is observed at SchnorrQ scheme, that is $22$ $\mu$s. This scheme should be preferred if the verification speed is of high importance. However recall that for our envisioned applications, the signer efficiency (energy efficiency) is of top priority and a small delay at the verifier is tolerable.

\section{Related Work}\label{sec:Related}%

There are two main lines of work to offer authentication for embedded medical devices: symmetric key primitives (e.g., MACs) and public key primitives (e.g., digital signatures). In this section, we only mention lightweight digital signature schemes that are most relevant to our work.

One-time signatures~(e.g., \cite{Yavuz:2013:EET:2462096.2462108,HORS_BetterthanBiBa02,ReiOneTime}) offer high computational efficiency, but usually have very large key and signature sizes that hinder their adoption in implantable medical devices. Moreover, they can only sign a pre-defined number of messages with a key pair, which introduce a key renewal overhead. The extensions of hash-based one-time signatures to multiple-time signatures (e.g., SPHINCS~\cite{SPHINCS:Bernstein:2015}) have high signing overhead, and therefore are not suitable for medical implantables. Some MAC based alternatives (e.g., TESLA \cite{ExtTESLA,ContiDelayedDisclosure}) use time asymmetries to offer computational efficient and compactness, they cannot offer non-repudiation and require a continuous time synchronization. EC-based digital signatures (e.g.,~\cite{FourQ,Ed25519,FourQ8bit,ECDSA,Conti_MedicalECCAuth,ARIS}) are currently the most prevalent alternatives to be used on embedded devices due to their compact size and higher signing efficiency compared to RSA-based signatures (e.g., CEDA~\cite{CEDA}).  We provided a detailed performance comparison of \uls~with its most recent EC-based alternatives in Section \ref{sec:PerformanceAnalysis}. 

\section{Conclusion}\label{sec:conclusion}

In this paper, we proposed \uls, that achieves the least energy consumption, the fastest signature generation along with the smallest signature among its ECC-based counterparts. \uls~is also immune to side-channel attacks aiming EC operations/exponentiations as well as to weak pseudo random number generators at the signer's side, since \uls~does not require any of these operations in its signature generation algorithm. We believe \uls~is highly preferable for applications wherein the signer efficiency is a paramount requirement, such as implantable medical devices. We implemented \uls~and its counterparts both on a resource-contrained device commonly used in medical devices and a commodity hardware.  Our experiments validate the significant energy efficiency and speed advantages of \uls~at the signer's side over its counterparts.\\ \vspace{-3mm}

\noindent \textbf{Acknowledgments.} This work is supported by the NSF Award \#1652389.

		\bibliographystyle{IEEEtran}
		\bibliography{crypto-etc}

\begin{thebibliography}{10}
\providecommand{\url}[1]{#1}
\csname url@samestyle\endcsname
\providecommand{\newblock}{\relax}
\providecommand{\bibinfo}[2]{#2}
\providecommand{\BIBentrySTDinterwordspacing}{\spaceskip=0pt\relax}
\providecommand{\BIBentryALTinterwordstretchfactor}{4}
\providecommand{\BIBentryALTinterwordspacing}{\spaceskip=\fontdimen2\font plus
\BIBentryALTinterwordstretchfactor\fontdimen3\font minus
  \fontdimen4\font\relax}
\providecommand{\BIBforeignlanguage}[2]{{%
\expandafter\ifx\csname l@#1\endcsname\relax
\typeout{** WARNING: IEEEtran.bst: No hyphenation pattern has been}%
\typeout{** loaded for the language `#1'. Using the pattern for}%
\typeout{** the default language instead.}%
\else
\language=\csname l@#1\endcsname
\fi
#2}}
\providecommand{\BIBdecl}{\relax}
\BIBdecl

\bibitem{MedicalDevice:SoK:SP:2014:Rushanan}
M.~Rushanan, A.~D. Rubin, D.~F. Kune, and C.~M. Swanson, ``Sok: Security and
  privacy in implantable medical devices and body area networks,'' in
  \emph{Proceedings of the 2014 IEEE Symposium on Security and Privacy}, ser.
  SP '14.\hskip 1em plus 0.5em minus 0.4em\relax IEEE Computer Society, 2014,
  pp. 524--539.

\bibitem{Ozmen_IOT_SP}
\BIBentryALTinterwordspacing
M.~O. Ozmen and A.~A. Yavuz, ``Low-cost standard public key cryptography
  services for wireless iot systems,'' in \emph{Proceedings of the 2017
  Workshop on Internet of Things Security and Privacy}, ser. IoTS\&P '17.\hskip
  1em plus 0.5em minus 0.4em\relax New York, NY, USA: ACM, 2017, pp. 65--70.
  [Online]. Available: \url{http://doi.acm.org/10.1145/3139937.3139940}
\BIBentrySTDinterwordspacing

\bibitem{MedicalDevice:Survey:2015:Camara2015272}
C.~Camara, P.~Peris-Lopez, and J.~E. Tapiador, ``Security and privacy issues in
  implantable medical devices: A comprehensive survey,'' \emph{Journal of
  Biomedical Informatics}, vol.~55, pp. 272 -- 289, 2015.

\bibitem{Repudiation:Survey}
\BIBentryALTinterwordspacing
M.~Vigil, J.~Buchmann, D.~Cabarcas, C.~Weinert, and A.~Wiesmaier, ``Integrity,
  authenticity, non-repudiation, and proof of existence for long-term
  archiving: A survey,'' \emph{Computers \& Security}, vol.~50, pp. 16 -- 32,
  2015. [Online]. Available:
  \url{http://www.sciencedirect.com/science/article/pii/S0167404814001849}
\BIBentrySTDinterwordspacing

\bibitem{FeasibilityIoT_Crypto}
A.~Ometov, P.~Masek, L.~Malina, R.~Florea, J.~Hosek, S.~Andreev, J.~Hajny,
  J.~Niutanen, and Y.~Koucheryavy, ``Feasibility characterization of
  cryptographic primitives for constrained (wearable) iot devices,'' in
  \emph{2016 IEEE International Conference on Pervasive Computing and
  Communication Workshops (PerCom Workshops)}, March 2016, pp. 1--6.

\bibitem{Yavuz:2013:EET:2462096.2462108}
A.~A. Yavuz, ``Eta: efficient and tiny and authentication for heterogeneous
  wireless systems,'' in \emph{Proceedings of the sixth ACM conference on
  Security and privacy in wireless and mobile networks}, ser. WiSec '13.\hskip
  1em plus 0.5em minus 0.4em\relax New York, NY, USA: ACM, 2013, pp. 67--72.

\bibitem{BPV:Ateniese:Journal:ACMTransEmbeddedSys:2017}
G.~Ateniese, G.~Bianchi, A.~T. Capossele, C.~Petrioli, and D.~Spenza,
  ``Low-cost standard signatures for energy-harvesting wireless sensor
  networks,'' \emph{ACM Trans. Embed. Comput. Syst.}, vol.~16, no.~3, pp.
  64:1--64:23, apr 2017.

\bibitem{BPV:basepaper:1998}
V.~Boyko, M.~Peinado, and R.~Venkatesan, ``Speeding up discrete log and
  factoring based schemes via precomputations,'' in \emph{Advances in
  Cryptology --- EUROCRYPT'98: International Conference on the Theory and
  Application of Cryptographic Techniques Espoo, Finland, May 31 -- June 4,
  1998 Proceedings}.\hskip 1em plus 0.5em minus 0.4em\relax Springer Berlin
  Heidelberg, 1998, pp. 221--235.

\bibitem{OfflineOnline_ImprovedShamir_2001}
A.~Shamir and Y.~Tauman, ``Improved online/offline signature schemes,'' in
  \emph{Proceedings of the 21st Annual International Cryptology Conference on
  Advances in Cryptology}, ser. CRYPTO '01.\hskip 1em plus 0.5em minus
  0.4em\relax London, UK: Springer-Verlag, 2001, pp. 355--367.

\bibitem{HighSpeedSignature:Bernstein:Journal2012}
D.~Bernstein, N.~Duif, T.~Lange, P.~Schwabe, and B.-Y. Yang, ``High-speed
  high-security signatures,'' \emph{Journal of Cryptographic Engineering},
  vol.~2, no.~2, pp. 77--89, 2012.

\bibitem{FourQ}
C.~Costello and P.~Longa, ``Four $\mathbb{Q}$ : Four-dimensional decompositions
  on a $\mathbb{Q}$ -curve over the mersenne prime,'' in \emph{Advances in
  Cryptology -- ASIACRYPT 2015}, T.~Iwata and J.~H. Cheon, Eds.\hskip 1em plus
  0.5em minus 0.4em\relax Springer Berlin Heidelberg, 2015, pp. 214--235.

\bibitem{RandomOracleModel93}
M.~Bellare and P.~Rogaway, ``Random oracles are practical: A paradigm for
  designing efficient protocols,'' in \emph{Proceedings of the 1st ACM
  conference on Computer and Communications Security ({CCS} '93)}.\hskip 1em
  plus 0.5em minus 0.4em\relax NY, USA: ACM, 1993, pp. 62--73.

\bibitem{Ed25519}
\BIBentryALTinterwordspacing
D.~J. Bernstein, N.~Duif, T.~Lange, P.~Schwabe, and B.-Y. Yang, ``High-speed
  high-security signatures,'' \emph{Journal of Cryptographic Engineering},
  vol.~2, no.~2, pp. 77--89, Sep 2012. [Online]. Available:
  \url{https://doi.org/10.1007/s13389-012-0027-1}
\BIBentrySTDinterwordspacing

\bibitem{DSAExponentiationCCS16}
C.~Pereida~Garc\'{\i}a, B.~B. Brumley, and Y.~Yarom, ``"make sure dsa signing
  exponentiations really are constant-time",'' in \emph{Proceedings of the 2016
  ACM SIGSAC Conference on Computer and Communications Security}, ser. CCS
  '16.\hskip 1em plus 0.5em minus 0.4em\relax New York, NY, USA: ACM, 2016, pp.
  1639--1650.

\bibitem{MedicalDevice:General:2011}
R.~R. Jueneman, ``Securing wireless medicine confidentiality, integrity,
  nonrepudiation, malware prevention,'' in \emph{Emerging Technologies for a
  Smarter World (CEWIT), 2011 8th International Conference Expo on}, Nov 2011,
  pp. 1--5.

\bibitem{Schnorr91}
C.~Schnorr, ``Efficient signature generation by smart cards,'' \emph{Journal of
  Cryptology}, vol.~4, no.~3, pp. 161--174, 1991.

\bibitem{GoldbergPIR}
I.~Goldberg, ``Improving the robustness of private information retrieval,'' in
  \emph{2007 IEEE Symposium on Security and Privacy (SP '07)}, 2007, pp.
  131--148.

\bibitem{JonathanKatzModernCrytoBook}
J.~Katz and Y.~Lindell, \emph{Introduction to Modern Cryptography}.\hskip 1em
  plus 0.5em minus 0.4em\relax Chapman \& Hall/CRC, 2007.

\bibitem{BPV:DistributionofmodularSum:2001:Nguyen}
I.~S. P.~Nguyen and J.~Stern, ``Distribution of modular sums and the security
  of the server aided exponentiation,'' in \emph{Proc. Workshop on Cryptography
  and Computational Number Theory (CCNT'99)}, vol.~20.\hskip 1em plus 0.5em
  minus 0.4em\relax Springer Berlin Heidelberg, pp. 257--268.

\bibitem{ATmega2560Medical}
P.~Szakacs-Simon, S.~A. Moraru, and F.~Neukart, ``Signal conditioning
  techniques for health monitoring devices,'' in \emph{2012 35th International
  Conference on Telecommunications and Signal Processing (TSP)}, July 2012, pp.
  610--614.

\bibitem{blakeHash}
\BIBentryALTinterwordspacing
J.-P. Aumasson, L.~Henzen, W.~Meier, and R.~C.-W. Phan, ``Sha-3 proposal
  blake,'' Submission to NIST (Round 3), 2010. [Online]. Available:
  \url{http://131002.net/blake/blake.pdf}
\BIBentrySTDinterwordspacing

\bibitem{BernsteinCHACHA}
\BIBentryALTinterwordspacing
D.~J. Bernstein, ``New stream cipher designs,'' M.~Robshaw and O.~Billet,
  Eds.\hskip 1em plus 0.5em minus 0.4em\relax Berlin, Heidelberg:
  Springer-Verlag, 2008, ch. The Salsa20 Family of Stream Ciphers, pp. 84--97.
  [Online]. Available: \url{http://dx.doi.org/10.1007/978-3-540-68351-3_8}
\BIBentrySTDinterwordspacing

\bibitem{Ed25519AVR}
M.~Hutter and P.~Schwabe, ``{NaCl} on 8-bit {AVR} microcontrollers,'' in
  \emph{Progress in Cryptology -- {AFRICACRYPT 2013}}, ser. Lecture Notes in
  Computer Science, vol. 7918.\hskip 1em plus 0.5em minus 0.4em\relax
  Springer-Verlag Berlin Heidelberg, 2013, pp. 156--172,
  \url{http://cryptojedi.org/papers/\#avrnacl}.

\bibitem{FourQ8bit}
Z.~Liu, P.~Longa, G.~C. C.~F. Pereira, O.~Reparaz, and H.~Seo,
  ``Four$\mathbb{Q}$ on embedded devices with strong countermeasures against
  side-channel attacks,'' in \emph{Cryptographic Hardware and Embedded Systems
  -- CHES 2017}, W.~Fischer and N.~Homma, Eds.\hskip 1em plus 0.5em minus
  0.4em\relax Cham: Springer International Publishing, 2017, pp. 665--686.

\bibitem{HORS_BetterthanBiBa02}
L.~Reyzin and N.~Reyzin, ``Better than {BiBa}: Short one-time signatures with
  fast signing and verifying,'' in \emph{Proceedings of the 7th Australian
  Conference on Information Security and Privacy ({ACIPS '02})}.\hskip 1em plus
  0.5em minus 0.4em\relax Springer-Verlag, 2002, pp. 144--153.

\bibitem{ReiOneTime}
K.~Kalach and R.~Safavi-Naini, ``An efficient post-quantum one-time signature
  scheme,'' in \emph{Selected Areas in Cryptography -- SAC 2015}, O.~Dunkelman
  and L.~Keliher, Eds.\hskip 1em plus 0.5em minus 0.4em\relax Cham: Springer
  International Publishing, 2016, pp. 331--351.

\bibitem{SPHINCS:Bernstein:2015}
D.~J. Bernstein, D.~Hopwood, A.~H{\"u}lsing, T.~Lange, R.~Niederhagen,
  L.~Papachristodoulou, M.~Schneider, P.~Schwabe, and Z.~Wilcox-O'Hearn,
  ``{SPHINCS}: Practical stateless hash-based signatures,'' in \emph{Advances
  in Cryptology -- EUROCRYPT 2015: 34th Annual International Conference on the
  Theory and Applications of Cryptographic Techniques}.\hskip 1em plus 0.5em
  minus 0.4em\relax Springer Berlin Heidelberg, April 2015, pp. 368--397.

\bibitem{ExtTESLA}
A.~Perrig, R.~Canetti, D.~Song, and D.~Tygar, ``Efficient and secure source
  authentication for multicast,'' in \emph{Proceedings of Network and
  Distributed System Security Symposium}, February 2001.

\bibitem{ContiDelayedDisclosure}
W.~B. Jaballah, M.~Conti, R.~D. Pietro, M.~Mosbah, and N.~V. Verde, ``Mass: An
  efficient and secure broadcast authentication scheme for resource constrained
  devices,'' in \emph{2013 International Conference on Risks and Security of
  Internet and Systems (CRiSIS)}, Oct 2013, pp. 1--6.

\bibitem{ECDSA}
\emph{{ANSI X9.62-1998:} Public Key Cryptography for the Financial Services
  Industry: The Elliptic Curve Digital Signature Algorithm ({ECDSA})}, American
  Bankers Association, 1999.

\bibitem{Conti_MedicalECCAuth}
M.~Wazid, A.~K. Das, N.~Kumar, M.~Conti, and A.~V. Vasilakos, ``A novel
  authentication and key agreement scheme for implantable medical devices
  deployment,'' \emph{IEEE Journal of Biomedical and Health Informatics},
  vol.~22, no.~4, pp. 1299--1309, July 2018.

\bibitem{ARIS}
R.~Behnia, M.~O. Ozmen, and A.~A. Yavuz, ``{ARIS:} authentication for
  {Real-Time} {IoT} systems,'' in \emph{2019 IEEE International Conference on
  Communications (ICC): Communication and Information Systems Security
  Symposium (IEEE ICC'19 - CISS Symposium)}, Shanghai, P.R. China, May 2019.

\bibitem{CEDA}
M.~O. {Ozmen}, R.~{Behnia}, and A.~A. {Yavuz}, ``Compact energy and delay-aware
  authentication,'' in \emph{2018 IEEE Conference on Communications and Network
  Security (CNS)}, May 2018, pp. 1--9.

\bibitem{SecProofSigScheme96Euro}
D.~Pointcheval and J.~Stern, ``Security proofs for signature schemes,'' in
  \emph{Proc. of the 15th International Conference on the Theory and
  Application of Cryptographic Techniques ({EUROCRYPT} '96)}.\hskip 1em plus
  0.5em minus 0.4em\relax Springer-Verlag, 1996, pp. 387--398.

\bibitem{Bellare-Neven:2006}
M.~Bellare and G.~Neven, ``Multi-signatures in the plain public-key model and a
  general forking lemma,'' in \emph{Proceedings of the 13th ACM Conference on
  Computer and Communications Security}, ser. CCS '06.\hskip 1em plus 0.5em
  minus 0.4em\relax New York, NY, USA: ACM, 2006, pp. 390--399.

\end{thebibliography}

\vspace{-1mm}
\appendix \label{sec:appendix}
\vspace{-1mm}

\noindent\textbf{\underline{Proof of Theorem~\ref{the:ULSSecurityTheorem}}} \\ \vspace{-3.5mm}

\noindent {\em Proof:}  If a polynomial-time  adversary \A~breaks the EU-CMA security of  \uls, then one can build another polynomially-bounded algorithm \F~that runs \A~as a subroutine and breaks \sch~signature. After setting $(y,Y) \leftarrow \schkg(1^{\kappa})$ and the corresponding parameters $(q,p,\alpha)$, \F~is run as in Definition \ref{def:EUCMA} as follows.  \F~handles  \A's and other queries in  Definition \ref{def:EUCMA}.

\vspace{1mm} \noindent \underline{{\em Algorithm $\mathcal{F}^{\schsig_{y}(\cdot)}(Y)$}}:

$\bullet$~\underline{{\em Setup:}} \F~maintains five lists \lm,~\lh,~\lhp,~\lhpp~and \lw, all initially empty. \lm~is a message list that records each message $M$ to be queried to \ulssig($\cdot$)~oracle by $ \mathcal{A} $. $\lh$,  $\lhp$ and $ \lhpp$~record the input $x$ to be queried to \ro~oracle and its corresponding \ro~answer $h$, respectively. $h \as\lhi (x)$, for $ i \in \{0,1,2\} $, returns the corresponding \ro~answer of $x$ stored in \lhi. \lw~keeps the record of messages that \F~queries to \schsig~oracle. \F~sets up \ro~and simulated public keys to initialize \ulssig~oracle as follows:

\begin{enumerate}[-]
	\setlength{\itemsep}{1pt}

	\item {\em Setup \ro~Oracle}: \F~implements a function \hsim~to handle \ro~queries. That is, the cryptographic hash function $H$ is modeled as a random oracle via \hsim~as follows. $h\as \hsim(x,\lhi)$: If $h\in \lhi$, for $ i \in \{0,1,2\} $,  then \hsim~returns the corresponding value $h\as  \lhi (M)$. Otherwise, it returns $h\Rq$ as the answer, inserts $(M,h)$ into $ \lhi$, respectively.
	
		\item {\em Setup Simulated Keys}: \F~selects parameters $(v,n)$ as in \ulskg$ (\cdot) $~Step 1, and works as follows:
	\begin{enumerate}[1)]
		\setlength{\itemsep}{1pt}
		\setlength{\parskip}{0pt}
		\setlength{\parsep}{0pt}
		\item Queries $\schsig_{y}(\cdot)$   on $w_j\Rq$, and receives signatures $(s'_j,h'_j)$, where $\{R'_j\as Y^{h'_j} \cdot \alpha^{s'_j} \bmod p\}_{j=0}^{n-1}$. 
		
		\item  Sets $\pk\as Y$ and system-wide public parameters $(v,n,q,p,\alpha)$ as in \ulskg\any.
 
		\item Sets $z\as \zo^{\kappa}$ and $\gamma=\{s'_j,h'_j\}_{j=0}^{n-1}$ and  $\beta=\{R'_{j}\}_{j=0}^{n-1}$ and inserts $\{w_j\}_{j=0}^{n-1}$ into \lw. \F~also inserts $\{w_j||R'_j,h'_j\}_{j=0}^{n-1}$ into $\lhpp$.		 \F~sets the counter $c\as 0$.
	\end{enumerate}
	
\end{enumerate}

\vspace{3pt}

$\bullet$~\underline{{\em Execute $(M^{*},\sigma^{*})\as \mathcal{A}^{\ro,\ulssig_{\sk}(\cdot)}(\pk)$}}: \F~handles \A's queries and forgery as follows:

\begin{enumerate}[-]
	\setlength{\itemsep}{1pt}
	\setlength{\parskip}{0pt}
	\setlength{\parsep}{0pt}
	
	\item  \underline{Queries of \A}: \A~queries \ro~and $\ulssig_{sk}(.)$ on any message of its choice up to $q_h$ and $\qs$ times, respectively. 
	
	\begin{enumerate}[1)]
		\itemsep 1pt
		\item {\em Handle \ro~queries}: \A~queries \ro~on a message $M$. \F~calls $h\as\hsim(M,\lhpp)$ and returns $h$ to \A.
		
		\item {\em Handle \ulssig$ (\cdot) $~queries}: \A~queries $\ulssig_{sk}\any$~on any message of its choice $M$. If $ M\in \lhp $ then \F~{\em aborts}. Otherwise, \F~inserts $M$ into \lm~and then continues as follows.
		\begin{enumerate}[i)]
			\setlength{\itemsep}{2pt}
			\setlength{\parskip}{0pt}
			\setlength{\parsep}{0pt}
			
			\item $x\as \hsim(z||c,\lh)$,   $(k_0,\ldots,k_{v-1})\as \hsim(z||x,\lhp)$, and fetch $\{s'_{k_i},h'_{k_i}\}_{i=0}^{v-1}$ from $\gamma$.
			
			\item $\oh\as\sum_{i=0}^{v-1} h'_{k_i} \bmod q$,  $\os\as\sum_{i=0}^{v-1} s'_{k_i} \bmod q$,   put $(M||x,\oh)$ in  $\lhpp$ and return $\sigma=(\os,x)$ to $ \mathcal{A} $.
			
		\end{enumerate}
	\end{enumerate}
	\vspace{1pt}
	
	\item \underline{Forgery of~\A}: Finally, $\mathcal{A}$ outputs a forgery for \pk~as $(M^{*},\sigma^{*})$, where $\sigma^{*}=(s^{*},x^*)$. By Definition \ref{def:EUCMA}, \A~wins the \EUCMA~experiment for \uls~if  $\ulsver(M^{*},\sigma^{*},\pk)=1$ and  $M^{*} \notin \lm$. 
	
\end{enumerate}

\vspace{3pt}
$\bullet$~\underline{{\em Forgery of \F}}:  If \A~fails, \F~also fails and returns $0$. Otherwise, given an \uls~forgery $(M^{*},\sigma^{*}=\langle s^{*},x^{*} \rangle)$ on \pk, \F~checks if  $M^{*}||x^{*} \notin \lhp$, or $x^{*} \notin \lh$   holds,  \F~{\em aborts}. Otherwise, using the forking lemma \cite{SecProofSigScheme96Euro,Bellare-Neven:2006}, \F~rewinds \A~with the same random tape, to get a new  forgery  	 $(M^{*},\tsigma =\langle  \tss,\txx \rangle)$. Given, both forgery pairs are valid, we have:   

	\setlength{\itemsep}{2pt}
	\setlength{\parskip}{0pt}
	\setlength{\parsep}{0pt}
	
\vspace{-5mm}
\begin{eqnarray*}
	R^{*} & \equiv & (\alpha^{y})^{h^{*}}\cdot \alpha^{s^{*}} \bmod p \\
	R^{*} & \equiv & (\alpha^{y})^{\thh}\cdot \alpha^{\tss} \bmod p
\end{eqnarray*}
These equations imply the below modular linear equations.
\begin{eqnarray*}
	r^{*} & \equiv & y\cdot h^{*} + s^{*} \bmod q \\
	r^{*} & \equiv & y\cdot \thh + \tss \bmod q
\end{eqnarray*}

\F~then extracts \sch~private key $y$ by solving the above modular linear equations (note that only unknowns are $y$ and $r^{*}$). \F~can further verify $Y \equiv \alpha^{y} \bmod p $. Given that \F~ extracted the \sch~private key, it is trivial to show that \F~can produce a valid forgery on any message $M^{**} \notin \lw $ of its choice on \sch~public key $Y$. Therefore, \F~wins \EUCMA~experiment for \sch~and returns 1.

\vspace{3pt}

\noindent\underline{\em{Probability Analysis  \& Transcript Indistinguishability}}: \F~may abort during simulation when \A~queries the \ulssig\any~oracle, if randomly chosen indexes $ (k_0,\dots,k_{v-1}) $ already exist in \lhp. The probability that happens is   $ \nicefrac{(q_h-1)q_s}{2^l} ,~l = v\cdot \log{n} $.  Based on \cite[Lemma 1]{Bellare-Neven:2006}, we define the success probability of  \A~as $ \texttt{ACC} $.  This probability  is defined as $ \texttt{ACC}  \geq \epsilon_\mathcal{A} - \nicefrac{(q_h-1)q_s}{2^l}  $. Where $ \epsilon_\mathcal{A} $ is the winning probability of  $ \mathcal{A} $.  The probability of \F~(i.e., $ \epsilon_\mathcal{F} $) for breaking \sch~is given by:
\vspace{-1mm}   
\begin{eqnarray*}
 \epsilon_\mathcal{F}  &  \geq  & \frac{\texttt{ACC} ^2}{q_h+q_s}-\frac{1}{2^{l}}\\
  & \geq & \frac{ \epsilon_{\mathcal{A}}^2 }{ (q_h+q_s)}-\frac{ 2((q_h-1)q_s)}{2^{l}(q_h+q_s)} 
\end{eqnarray*}

 $ \mathcal{A} $'s view in Algorithm \ref{alg:ULS_Generic} is the public key $ \pk $, signatures $(\sigma_1,\dots, \sigma_{q_s -1}) $ and hash outputs. The public key in the simulation has the identical distribution as the one in original \uls~- they are both the output of the \schkg\any~ algorithm. As for the signatures,  $ \sigma=( s,x ) $, both elements $ s $ and $ x $ have the same distribution as in the original scheme.

\end{document}